%% file: main.tex
\documentclass[submission,copyright,creativecommons]{eptcs}
\providecommand{\event}{FROM 2019} 
\usepackage{breakurl}             
\usepackage{underscore}           
\usepackage{array}
\usepackage{amsthm}
\usepackage[english]{babel}
\usepackage{graphicx}
\usepackage{listings}
\usepackage{dafny}
\usepackage[vlined]{algorithm2e}
\usepackage{import}
\usepackage[inline]{enumitem}

\theoremstyle{definition}

\newtheorem{example}{Example}
\newtheorem{predicate}{Predicate}
\newtheorem{invariant}{Invariant}
\newtheorem{lemma}{Lemma}

\usepackage[utf8]{inputenc}

\usepackage{tikz}
\usetikzlibrary{shapes.geometric}
\usetikzlibrary{arrows.meta}

\lstnewenvironment{datastructure}[1][]{%
  \lstset{language=dafny,#1}}
  {}

\newenvironment{figalgorithm}[1][htbp]
  {\begin{figure}[#1]}
  {\end{figure}}

\newcommand{\Csh}{C\#}

\title{Verifying the DPLL Algorithm in Dafny}
\author{Cezar-Constantin Andrici
\institute{Alexandru Ioan Cuza University of Ia\c{s}i}
\email{cezar.andrici@gmail.com}
\and
\c{S}tefan Ciob\^ac\u{a}
\institute{Alexandru Ioan Cuza University of Ia\c{s}i}
\email{stefan.ciobaca@info.uaic.ro}
}

\def\titlerunning{Verifing the DPLL algorithm in Dafny}

\def\authorrunning{C.C. Andrici \& \c{S}. Ciob\^{a}c\u{a}}
\begin{document}
\maketitle

\begin{abstract}
  
  Modern high-performance SAT solvers quickly solve large
  satisfiability instances that occur in practice. If the instance is
  satisfiable, then the SAT solver can provide a witness which can be
  checked independently in the form of a satisfying truth assignment.

  However, if the instance is unsatisfiable, the certificates could be
  exponentially large or the SAT solver might not be able to output
  certificates. The implementation of the SAT solver should then be
  trusted not to contain bugs. However, the data structures and
  algorithms implemented by a typical high-performance SAT solver are
  complex enough to allow for subtle programming errors.

  To counter this issue, we build a verified SAT solver using the
  Dafny system. We discuss its implementation in the present article.
  
\end{abstract}

\section{Introduction}

Recent advances have enabled the formal verification and analysis of
larger and larger software projects using techniques such static
analysis or automated or interactive theorem proving. Two such
examples are the certified C compiler CompCert~\cite{compcert} (using
the Coq proof assistant) and the seL4 microkernel~\cite{sel4} (using
Isabelle/HOL).

However, one problematic aspect is that verification tools themselves
(e.g., static analyzers, SAT/SMT solvers, theorem provers,
verification condition generators), which are usually highly
sophisticated and relatively large pieces of software, might also
include bugs. Such bugs can be rather embarassing, as they may lead
the verification tool to prove correct a wrong program.

In fact, many verification tools contain bugs. For example, Brummayer
and others~\cite{DBLP:conf/sat/BrummayerLB10} have shown using fuzz
testing that many state-of-the-art SAT solvers contained bugs,
including soundness bugs. Since 2016, in order to mitigate this issue,
the anual SAT competiton requires solvers competing in the main track
to output UNSAT certificates~\cite{DBLP:conf/aaai/BalyoHJ17}; these
certificates are independently checked in order to ensure
soundness. Other contents such as the Automated Theorem Proving
competition~\cite{DBLP:journals/aicom/Sutcliffe18} or SMT
competition~\cite{DBLP:journals/jsat/CokDW14} also contain various
soundness check at various points in the competition timeline.

An approach to help ensuring correctness is to \emph{verify the
  verification tools} themselves. In this article, we propose to do
just that for a SAT solver. A SAT solver solves instances of the
well-known Boolean satisfiability problem (SAT), which has many
applications in software and hardware verification, as well as in
combinatorial optimization. Relatively recently, high-performance SAT
solvers based on the DPLL and CDCL algorithms have emerged and can
handle practical SAT instances with millions of variables in
reasonable running time.

However, SAT solver implementations contain complex data structures
and algorithms and can therefore contain subtle bugs. This is less of
an issue for satisfiable instances, since a satisfiability certificate
can be checked easily. To counter this possible soundness issue, the
SAT Competition started to require solvers to output certificates even
in the case when the formula is unsatisfiable. Unsurprisingly, these
certificates can be exponential and some tools cannot output
certificates.

In this paper, we propose to instead implement a \emph{verified SAT
  solver} using the Dafny system. Dafny is a high-level imperative
language with support for object oriented features. Dafny features
methods with preconditions, postconditions and invariants which are
checked at compilation time by relying on the Z3 SMT solver. If a
postcondition cannot be established (either due to a timeout or due to
the fact that it does not hold), compilation fails. Therefore, we can
place a high degree of trust in a program verified using the Dafny
system.

A modern, high-performance SAT solver essentially consists of the
following optimizations over a backtracking approach to solving the
satisfiability question:
\begin{enumerate*}
\item unit propagation;
\item fast data structures to identify unit clauses;
\item variable ordering heuristics;
\item backjumping;
\item conflict analysis;
\item clause learning;
\item restart strategy.
\end{enumerate*}

The first three items are usually refered to as the DPLL algorithm,
while the last four items make up the CDCL algorithm. We have
implemented and verified the first two items, and the other points
remain for future work. In particular, we use a fixed variable
ordering. We note that our verification ensures soundness,
completeness and termination of the solver. We do not verify input
handling.

In Section~\ref{sec:related}, we discuss related work. In
Section~\ref{sec:dpll}, we briefly explain the DPLL algorithm, as
presented in the literature. In Section~\ref{sec:verified}, we present
the Dafny data structures and their invariants, as well as the
implementation of the core algorithm, together with the verified
guarantees that it provides. In Section~\ref{sec:benchmarks}, we
briefly benchmark the performance of our solver and we conclude in
Section~\ref{sec:conclusion}.

\section{Related Work}
\label{sec:related}

The SAT solver {\tt versat}~\cite{DBLP:conf/vmcai/OeSOC12} was
implemented and verified in the Guru programming language using
dependent types. As our solver, it also implements efficient data
structures. However, it relies on a translation to C where data
structures are implemented imperatively by using reference counting
and a statically enforced read/write discipline. Unlike our approach,
the solver is only verified to be sound: if it produces an {\tt UNSAT}
answer, then the input formula truly is unsatisfible. However,
termination and completeness (if the solver produces SAT, then the
formula truly is satisfiable) are not verified. Another small
difference is the verification guarantee: {\tt versat} is verified to
output {\tt UNSAT} only if a resolution proof of the empty clause
exists, while in our approach we use a semantic criterion: our solver
always terminates and produces UNSAT only if there is no satisfying
model of the input formula. Of course, in the case of propositional
logic these criteriors are equivalent and therefore this difference is
mostly a matter of implementation. Unlike our solver, some checks are
not proved statically and must be checked dynamically, so they could
be a source of incompleteness. An advantage of {\tt versat} over our
approach is that is implements more optimizations, like conflict
analysis and clauses learning, which enable it to be
competitive. Blanchette and
others~\cite{DBLP:journals/jar/BlanchetteFLW18} present a certified
SAT solving framework verified in the Isabelle/HOL proof
assistant. The proof effort is part of the \emph{Isabelle
  Formalization of Logic} project. The framework is based on
refinment: at the highest level sit several calculi like CDCL and
DPLL, which are formally proved. Depending on the strategy, the
calculi are also shown to be terminating. The calculi are shown to be
refined by a functional program. Finally, at the lowest level is an
imperative implementation in Standard ML, which is shown to be a
refinement of the functional implementation. Emphasis is also placed
on meta-theoretical consideration. The final solver can still two
orders of magnitude slower than a state-of-the-art C solver and
therefore additional optimizations~\cite{DBLP:conf/nfm/Fleury19} are
desirable. In contrast, in our own work we do not investigate any
metatheoretical properties of the DPLL/CDCL frameworks; we simply
concentrate on obtaining a verified SAT solver. We investigate to what
extent directly proving the imperative algorithm is possible in an
autoactive manner. We have shown that this is possible for a
restricted algorithm. However, we have reached a point where Dafny
proofs take a lot of time (tens of minutes). In order to go further
and verify the entire DPLL algorithm, additional techniques to bring
down Dafny verification time are required, as dicussed in the
conclusion. Another SAT solver formalized in Isabelle/HOL is by
Mari\'{c}~\cite{DBLP:journals/jar/Maric09}. In contrast to the
previous formalization, the verification methodology is not based on
refinement. Instead, a shallow embedding of the algorithm is expressed
as a set of recursive functions. This style of algorithm is of course
not as high-performance as an imperative one. Another formalization of
a SAT solver (extended with linear arithmetic) is by
Lescuyer~\cite{lescuyer:tel-00713668}, who verifies a DPLL-based
decision procedure for propositional logic in Coq and exposes it as a
reflexive tactic. Finally, a decision procedure based on DPLL is also
verified by Shankar and Vaucher~\cite{DBLP:journals/entcs/ShankarV11}
in the PVS system. For the proof, they rely on subtyping and dependent
types.


\section{The Davis-Putnam-Logemann-Loveland Algorithm}
\label{sec:dpll}

The DPLL procedure is an optimization of backtracking. The main
improvement is called unit propagation. A unit clause has the property
that its literals are all \textit{false} except one, which has no
value yet. If this literal would be set to \textit{false}, the clause
would not be satisfied; therefore, the literal must necessarily be
\textit{true}. This process of identifying unit clauses and settings
the unknown literal to true is called unit propagation.

\begin{example}\label{theexample} We consider a formula with 7 variables and 5 
clauses: 
\begin{center}
    $(x_{1} \lor x_{2} \lor x_{3}) \land$ 
    $(\lnot x_{1} \lor \lnot x_{2}) \land$
    $(x_{2} \lor \lnot x_{3}) \land $
    $(x_{2} \lor x_{4} \lor x_{5}) \land$
    $(x_{5} \lor x_{6} \lor x_{7})$
\end{center}
The formula is satisfiable, as witness by the truth assignment \textit{(true, false, false, true, true, false, true)}.
\end{example}

Algorithm~\ref{dpllrecursive} describes the DPLL procedure which we
implement and verify.  We describe how the algorithm works on this
example: first, the algorithm choses the literal $x_{1}$ and sets it
to \textit{true} (arbitrarily; if \textit{true} would not work out,
then the algorithm would backtrack here and try \textit{false}). At
the next step, it finds that the second clause is unit and sets $\lnot
x_{2}$ to \textit{true}, which makes the third clause unit, so $\lnot
x_{3}$ is set to \textit{true}. After unit propagation, the next
clause not satisfied yet is the fourth one, and the first unset
literal is $x_{4}$. At the branching step, $x_{4}$ is assigned to
\textit{true}.  Furthermore, only one clause is not satisfied yet, and
the next decision is to choose $x_{5}$ and set it to \textit{true},
which makes the formula satisfied, even if $x_{6}$ and $x_{7}$ are not
set yet.

\begin{figalgorithm}
\begin{algorithm}[H]
\SetKwProg{main}{Function}{}{}
\SetKwFunction{funone}{DPLL-recursive}

\main{
    \funone{$F$, $\textit{tau}$}}{
    \SetKwInOut{Input}{input}\SetKwInOut{Output}{output}
    \Input{A CNF formula $F$ and an partial assignment $\textit{tau}$}
    \Output{UNSAT, or an assignment satisfying $F$}
    \While{$\exists$ unit clause $\in F$}{
        $\ell \leftarrow$ the unset literal from the unit clause \\
        $\textit{tau} \leftarrow \textit{tau}[\ell := true]$
    }
    
    \lIf{F contains the empty clause}{\Return{UNSAT}}
    \If{F has no clauses left} {
        Output $\textit{tau}$ \\
        \Return{SAT}
    }
    
    $\ell \leftarrow$ first unset literal that appears in the first not satisfied clause
    
    \lIf{$\textit{DPLL-recursive}(F, \textit{tau}[\ell := \textit{true}]) = \textit{SAT}$}{\Return{$\textit{SAT}$}}
    
    \Return{$\textit{DPLL-recursive}(F, \textit{tau}[\ell := \textit{false}])$}
}
\end{algorithm}
\caption{Presented in Satisfiability solvers, 2017~\cite{DBLP:reference/fai/GomesKSS08}, slightly modified in order to match our implementation.}
\label{dpllrecursive}
\end{figalgorithm}
\setcounter{figure}{0}

\section{A Verified Implementation of the DPLL Algorithm}

In this section, we describe the structure of our verified solver.

\label{sec:verified}

\subsection{Data Structures and Their Invariants}

We have 2 classes, $\textit{Stack}$ and $\textit{Formula}$. For both
of them and their data structures we formulate the conditions that
hold before and after each step that modifies them.

\begin{figure}[ht]
\begin{datastructure}
class Stack {
  var size : int; 
  var stack : array< seq<(int, bool)> >;
  var variablesCount : int;
  ghost var contents : set<(int, bool)>;
}
\end{datastructure}
\caption{Data Structure - Stack.}
\end{figure}

The attribute $\textit{variablesCount}$ is the number of variables in the formula.
The array \textit{stack} has size \textit{variablesCount} and contains at most this many \emph{layers}, which are explained below.
A propositional variable is represented in the stack by a value between $0$
and $\textit{variablesCount}-1$.
The attribute \textit{size} represents the number of layers on the \textit{stack}.

The \textit{stack} contains the trail of assignments made up to the
current state, divided into layers. A new layer is created at every
branching step, where a new unset variable $v$ is chosen and set. The
first element is $(v, \textit{boolean value})$; the rest of the
sequence in the layer contains assignments performed by unit
propagation. In this way, every time the algorithm backtracks it knows
exactly how many assignments to revert to reach the previous state. An
instance of the stack is shown in Figure~\ref{fig:stackpresentation}
for Example~\ref{theexample}, where it can be seen that after the
first iteration, the variable chosen to be set was $x_{1}$, and the
variables set by the unit propagation were $x_{2}$ and $x_{3}$.

\begin{figure}[ht]
  \begin{center}
    \input{example-figure.tex}
  \end{center}
\caption{Stack representation for Example~\ref{theexample}.}
\label{fig:stackpresentation}
\end{figure}

Ghost constructs are used only during
verification~\cite{dafnyReferenceManual}. We use this feature for
\textit{contents}, which is a variable that makes easier to implement
and prove various conditions about the content of the
\textit{stack}. It has exactly the same content as \textit{stack}, but
it is stored as a set, where order does not matter.


To represent class invariants, Dafny encourages a methodology of defining
a \textbf{valid()} predicate, which is used as a pre-condition for all
class members. In our case \textbf{valid()} is a conjunction of the following
invariants.

\begin{invariant}
  \textit{Stack} contains assignments only on the used layers.
$\forall i \bullet 0 <= i < size-1 \Longrightarrow |stack[i]| > 0
\land \forall i \bullet size <= i < |stack| \Longrightarrow |stack[i]| == 0$.
\end{invariant}

\begin{invariant}
Each variable occurs at most once in the stack.
    \begin{dafny}
    forall i,j :: 0 <= i < |stack| &&  0 <= j < |stack[i]| ==>
        (forall i', j' :: i < i' < |stack| && 0 <= j' < |stack[i']| ==> 
            stack[i][j].0 != stack[i'][j'].0)
        (forall j' :: j < j' < |stack[i]| ==> stack[i][j].0 != stack[i][j'].0))
    \end{dafny}
\end{invariant}

\begin{invariant}
Every assignment which occurs in the stack also occurs in the ghost var $\textit{contents}$.

\begin{dafny}
(forall i, j :: 0 <= i < |stack| && 0 <= j < |stack[i]| ==> stack[i][j] in contents) &&
(forall c :: c in contents ==> 
    exists i,j :: 0 <= i < stack.Length && 0 <= j < |stack[i]| && stack[i][j] == c)
\end{dafny}
\end{invariant}

\begin{figure}[ht]
\begin{datastructure}
class Formula {
  var variablesCount : int;
  var clauses : seq< seq<int> >;
  var stack : Stack;
    
  var truthAssignment : array<int>;
    
  var trueLiteralsCount : array<int>;
  var falseLiteralsCount : array<int>;
   
  var positiveLiteralsToClauses : array< seq<int> >;
  var negativeLiteralsToClauses : array< seq<int> >;
}
\end{datastructure}
\caption{Data Structure - Formula.}\label{formulaDef}
\end{figure}

A \textbf{Formula} is a tuple $(\textit{variablesCount},
\textit{clauses}, \textit{stack})$.  Because a variable is represented
by a value between $0$ and $\textit{variablesCount}-1$, a positive
literal is denoted in \textit{clauses} by $\textit{variable}+1$ and a
negative literal by $-\textit{variable}-1$.  Based on the tuple, we
have created 5 more efficient data structures that contain the same
information, which have the following invariants:

The array \textbf{truthAssignment} is indexed from $0$ to
$\textit{variablesCount}-1$, where $\textit{truthAssignment}[v]$ means that for
the current state the variable $v$ has the value: $-1$ if unset, $0$
if false, $1$ if true. It is created based on the data structure
\textit{stack} and is updated every time \textit{stack} is
updated. Looking at Invariant~\ref{invariantTruthAssignment}, it is
easy to see how simple \textit{truthAssignment} is defined by using
the ghost variable $stack.contents$. If it would have been defined
based on the layers, several additional universal quantifiers would
have been needed. The function $\textit{getLiteralValue}(\textit{tau},
\ell)$ returns the value of the literal $\ell$ in the truth assignment
$\textit{tau}$. The invariant for \textbf{truthAssignment} is:

\begin{invariant}\label{invariantTruthAssignment}
$\textit{validTruthAssignment}()$
\begin{dafny}
|truthAssignment| == variablesCount &&
(forall i :: 0 <= i < |truthAssignment| ==> -1 <= truthAssignment[i] <= 1) &&
(forall i :: 0 <= i < |truthAssignment| && truthAssignment[i] != -1 ==>
  (i, truthAssignment[i]) in stack.contents) &&
(forall i :: 0 <= i < |truthAssignment| && truthAssignment[i] == -1 ==>
  (i, false) !in stack.contents && (i, true) !in stack.contents)
    \end{dafny}
\end{invariant}

The variables \textbf{trueLiteralsCount} and
\textbf{falseLiteralsCount} are two arrays indexed from $0$ to
$|\textit{clauses}|-1$, where $\textit{trueLiteralsCount}[i]$ denotes
the number of literals set to true in $\textit{clause}_{i}$ and
$\textit{falseLiteralsCount}[i]$ the number of false literals in
$\textit{clause}_{i}$. These are used to quickly identify which
clauses are satisfied, which clauses are unit or which clauses are
false. For example, to check whether a $clause_{i}$ is satisfied, we
only evaluate $\textit{trueLiteralsCount}[i] > 0$. The following
invariants are true for these arrays:

\begin{invariant}
$\textit{validTrueLiteralsCount}()$ (analogously for
  $\textit{validFalseLiteralsCount}()$)

    \begin{dafny}
|trueLiteralsCount| == |clauses| &&
forall i :: 0 <= i < |clauses| ==>
  0 <= trueLiteralsCount[i] == countTrueLiterals(truthAssignment, clauses[i])
    \end{dafny}
\end{invariant}

The arrays \textbf{positiveLiteralsToClauses} and
\textbf{negativeLiteralsToClauses} are two arrays indexed from $0$ to
$\textit{variablesCount}-1$, where
$\textit{positiveLiteralsToClauses}[i]$ contains the indices of the
clauses where a given variable occurs and
$\textit{negativeLiteralsToClauses}[i]$ the indices of the clauses
where its negation occurs. These data structures are used every time a
new literal is set/unset in order to update \textit{trueLiteralsCount}
and \textit{falseLiteralsCount} and to do unit propagation. They
satisfy the following invariants:

\begin{invariant}
$\textit{validPositiveLiteralsToClauses}()$ (analogously for
  $\textit{validNegativeLiteralsToClauses}()$)

    \begin{dafny}
|positiveLiteralsToClauses| == variablesCount && (
    forall variable :: 0 <= variable < |positiveLiteralsToClauses| ==>
        ghost var s := positiveLiteralsToClauses[variable];
        valuesBoundedBy(s, 0, |clauses|) && orderedAsc(s) &&
        (forall clauseIndex :: clauseIndex in s ==> variable+1 in clauses[clauseIndex]) &&
        // the clauses which do not appear, do not contain the positive literal
        (forall clauseIndex :: 0 <= clauseIndex < |clauses| && clauseIndex !in s ==>
            variable+1 !in clauses[clauseIndex]))
    \end{dafny}
\end{invariant}

The conjunction of the above invariants, plus a few other low-level
predicates that we omit for brevity, are incorporated in a single
predicate \textit{valid()} which is used as a data structure invariant
for all methods. This way, it is guaranteed that the data structures
are consistent.

\subsection{Proof of the Data Structure Invariants}\label{correctnessCode}

From the initial valid state, we can to do one of four actions: create
a new layer on the stack, set a variable, set a literal and do unit
propagation, and undo the last layer on the stack. We show that these
four methods preserve the data structure invariants above.

The method \textbf{newLayerOnStack()} increments the size of the stack
by one, but it has the following preconditions: the stack must not be
full and the last layer must not be empty. The method guarantees that
the new state is valid, and nothing changes except the size of the
stack.
\begin{dafny}
method newLayerOnStack()
    requires valid();
    requires 0 <= stack.size < |stack.stack|;
    requires stack.size > 0 ==> |stack.stack[stack.size-1]| > 0;

    modifies stack;

    ensures valid();
    ensures stack.size == old(stack.size) + 1;
    ensures 0 < stack.size <= |stack.stack|;
    ensures forall i :: 0 <= i < |stack.stack| ==> stack.stack[i] == old(stack.stack[i]);
    ensures stack.contents == old(stack.contents);
\end{dafny}

The method \textbf{setVariable(variable : int, value : bool)} requires
a variable that is not set in the current valid state and guarantees
that only one position in the new truth assignment was changed: the
position for \textit{variable}. Because \textit{stack.stack} and
\textit{truthAssignment} were changed, \textit{trueLiteralsCount} and
\textit{falseLiteralsCount} have to be updated. We use
\textit{positiveLiteralsToClauses} and
\textit{negativeLiteralsToClauses} to efficiently update them, and
prove that the ones that are not contained in those are not
impacted. To prove termination, we use as a variant the number of
unset variables, which decreases at every branching step of the
algorithm.
\begin{dafny}
method setVariable(variable : int, value : bool)
    requires valid();
    requires 0 <= variable < variablesCount;
    requires truthAssignment[variable] == -1;
    requires 0 < stack.size <= |stack.stack|;

    modifies truthAssignment, stack, stack.stack, trueLiteralsCount,
             falseLiteralsCount;

    ensures valid();
    ensures stack.size == old(stack.size);
    ensures 0 < stack.size <= |stack.stack|;
    ensures |stack.stack[stack.size-1]| == |old(stack.stack[stack.size-1]| + 1;
    ensures stack.contents == old(stack.contents) + {(variable, value)};
    ensures forall i :: 0 <= i < |stack.stack| && i != stack.size-1 ==>
                stack.stack[i] == old(stack.stack[i]);
    ensures value == false ==> old(truthAssignment[variable := 0]) == truthAssignment;
    ensures value == true ==> old(truthAssignment[variable := 1]) == truthAssignment;
    ensures countUnsetVariables(truthAssignment) + 1 ==
                countUnsetVariables(old(truthAssignment));
\end{dafny}      

The method \textbf{setLiteral(literal : int, value : bool)} uses
\textit{setVariable}, so the pre- and post-conditions are similar, but
the difference is that after it sets the first literal, it also
performs unit propagation. This means that it calls
\textit{setLiteral} again with new values. So, at the end of a call,
the \textit{truthAssignment} might change at several positions.
\begin{dafny}
method setLiteral(literal : int, value : bool)
    requires valid();
    requires validLiteral(literal);
    requires getLiteralValue(truthAssignment, literal) == -1;
    requires 0 < stack.size <= |stack.stack|;

    modifies truthAssignment, stack, stack.stack, trueLiteralsCount,
             falseLiteralsCount;

    ensures valid();
    ensures 0 < stack.size <= |stack.stack|;
    ensures stack.size == old(stack.size);
    ensures |stack.stack[stack.size-1]| > 0;
    ensures forall i :: 0 <= i < |stack.stack| && i != stack.size-1 ==>
                stack.stack[i] == old(stack.stack[i]);
    ensures forall x :: x in old(stack.contents) ==> x in stack.contents;
    ensures stack.contents == old(stack.contents) + stack.getLastLayer();
    ensures countUnsetVariables(truthAssignment) <
                old(countUnsetVariables(truthAssignment));
    ensures isSatisfiableExtend(oldTau[literal := value]) <==> 
              isSatisfiableExtend(truthAssignment);
    decreases countUnsetVariables(truthAssignment);
\end{dafny}

Finally, the method \textbf{undoLayerOnStack()} reverts the
assignments from the last layer by changing the value of the literals
to \textit{unset}. As \textit{setVariable}, this method needs several
proofs that confirm that the data structures are updated correctly and
that the state is valid. To quickly update \textit{trueLiteralsCount}
and \textit{falseLiteralsCount}, we used
\textit{positiveLiteralsToClauses} and
\textit{negativeLiteralsToClauses}, and proved that the ones that are
not in those remain unchanged.
\begin{dafny}
method undoLayerOnStack()
    requires valid();
    requires 0 < stack.size <= |stack.stack|;
    requires |stack.stack[stack.size-1]| > 0;

    modifies truthAssignment, stack, stack.stack, trueLiteralsCount,
             falseLiteralsCount;

    ensures valid();
    ensures stack.size == old(stack.size) - 1;
    ensures 0 <= stack.size < |stack.stack|;

    ensures forall i :: 0 <= i < |stack.stack| && i != stack.size ==>
                stack.stack[i] == old(stack.stack[i]);
    ensures |stack.stack[stack.size]| == 0;
    ensures forall x :: x in old(stack.contents) && x !in old(stack.stack[stack.size-1]) ==>
                x in stack.contents;
    ensures forall x :: x in old(stack.stack[stack.size-1]) ==> x !in stack.contents;
    ensures stack.contents == old(stack.contents) - old(stack.getLastLayer());
    ensures stack.size > 0 ==> |stack.stack[stack.size-1]| > 0;

\end{dafny}

\subsection{Proof of Functional Correctness}
\label{correctnessResult}

The entry point called to solve the SAT instance is
\textit{solve}: \begin{dafny}
method solve() returns (result : SAT_UNSAT)
  requires valid();
  requires 0 <= formula.stack.size <= formula.stack.stack.Length;
  requires formula.stack.size > 0 ==>
    |formula.stack.stack[formula.stack.size-1]| > 0;

  modifies formula.truthAssignment, formula.stack, formula.stack.stack,
             formula.trueLiteralsCount, formula.falseLiteralsCount;

  ensures valid();
  ensures old(formula.stack.size) == formula.stack.size;
  ensures forall i :: 0 <= i < |formula.stack.stack| ==>
    formula.stack.stack[i] == old(formula.stack.stack[i]);
  ensures old(formula.stack.contents) == formula.stack.contents;
  ensures formula.stack.size > 0 ==>
    |formula.stack.stack[formula.stack.size-1]| > 0;

  ensures result.SAT? ==> formula.isSatisfiableExtend(formula.truthAssignment);
  ensures result.UNSAT? ==>
    !formula.isSatisfiableExtend(formula.truthAssignment);

  decreases countUnsetVariables(formula.truthAssignment);
\end{dafny}

It implements the $\textit{DPLL-procedure}$ using recursion, but the data
structures are kept in the instance of a class instead of being passed
as arguments. The pre- and post-conditions of \textit{solve} can be
summed up by the following: it starts in a valid state and it ends up
in the exact same state, and if it returns SAT then the current
\textit{truthAssignment} can be extended to satisfy the formula, and
if returns UNSAT it means that no truth assignment extending the
current \textit{truthAssignment} satisfies it. We made use of the
following predicate:

\begin{predicate}
  \label{satExt}
  
  $\textit{isSatisfiableExtend}(\textit{tau}, \textit{clauses})$: A
  set of clauses are satisfiable by a partial truth assignment
  \textit{tau} if there exists a complete assignment that extends
  $\textit{tau}$ and that satisfies the formula.
  
\end{predicate}

Starting and ending in the same state means that we chose to undo the
changes we made even if we found a solution, and this is because
otherwise we would have had to add a condition to every
\textit{ensures} clause with the type of the result, which would have
doubled the number of post-conditions. For simplicity, we chose to
revert to the initial state every time.

A flowchart that shows every command in the \textit{solve} method, and
the propositions that hold after each line, is presented in
Figure~\ref{fig:flowchartsolve}. For simplicity, when the initial
state is reached, we used the notation $\textit{state} = \textit{old(state)}$.

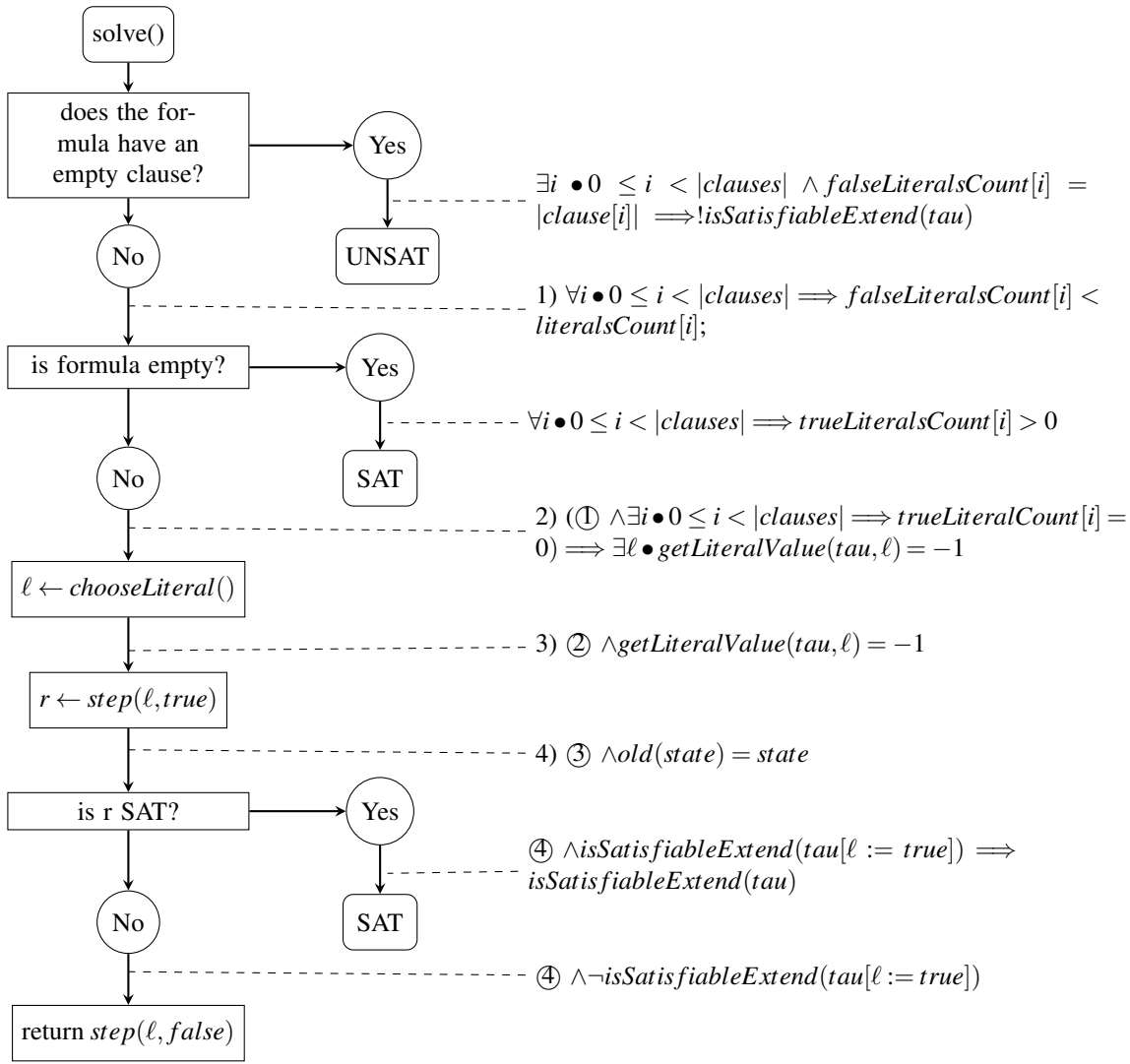
\begin{figure}[ht]
\input{solve-flowchart.tex}
\caption{Flowchart of method \textit{solve}.}\label{fig:flowchartsolve}
\end{figure}

The modifications are extracted to the method $step(\ell, value)$ to 
be easier to prove that we modify the data structures and that at the 
end we revert the changes to reach to the initial state. Most of the 
pre- and post-conditions are exactly the same as the ones in \textit{solve}, 
with small differences. First, \textit{step} takes an unset literal and 
returns SAT if $isSatisfiableExtend(tau[\ell := value])$ and UNSAT if not. 
With these \textit{ensures} clauses, \textit{solve} can find a solution or 
prove using Lemma~\ref{forVariableNotSatisfiableExtend_notSatisfiableExtend} 
that the current truth assignment could not be extended to satisfy the formula.

By putting together the methods 
described in Section~\ref{correctnessCode} it is easy to see how they fit and 
how the proof is build.

\begin{dafny}
method step(literal : int, value : bool) returns (result : SAT_UNSAT)
    requires valid();
    requires 0 <= formula.stack.size < |formula.stack.stack|;
    requires formula.stack.size > 0 ==>
      |formula.stack.stack[formula.stack.size-1]| > 0;
    requires !formula.hasEmptyClause();
    requires !formula.isEmpty();
    requires formula.validLiteral(literal);
    requires formula.getLiteralValue(formula.truthAssignment, literal) == -1;

    modifies formula.truthAssignment, formula.stack, formula.stack.stack,
             formula.trueLiteralsCount, formula.falseLiteralsCount;

    ensures valid();
    ensures old(formula.stack.size) == formula.stack.size;
    ensures forall i :: 0 <= i < |formula.stack.stack| ==>
      formula.stack.stack[i] == old(formula.stack.stack[i]);
    ensures old(formula.stack.contents) == formula.stack.contents;
    ensures formula.stack.size > 0 ==>
            |formula.stack.stack[formula.stack.size-1]| > 0;
    ensures result.SAT? ==>
      formula.isSatisfiableExtend(formula.truthAssignment[literal := value]);
    ensures result.UNSAT? ==>
      !formula.isSatisfiableExtend(formula.truthAssignment[literal := value]);
{
  formula.newLayerOnStack();
  // stack.size == old(stack.size) + 1
  ghost var tau' := formula.truthAssignment[literal := value];
  formula.setLiteral(literal, value);
  // isSatisfiableExtend(tau') <==> isSatisfiableExtend(formula.truthAssignment)
  result := solve();
  // isSatisfiableExtend(formula.truthAssignment) || 
  //  !isSatisfiableExtend(formula.truthAssignment)
  formula.undoLayerOnStack();
  // old(state) == state && (isSatisfiableExtend(tau') || !isSatisfiableExtend(tau'))
  return result;
}
\end{dafny}

The method $\textit{setLiteral}(\ell, \textit{value})$ also ensures
that the formula can be satisfiable under the returned truth
assignment (let us denote it by \textit{finalTau}) if and only if it
can be satisfiable under the initial truth assignment with $\ell
\leftarrow value$ (\textit{tau}). The difference between \textit{tau}
and \textit{finalTau} is that \textit{finalTau} also contain the
assignments performed during unit propagation.

To do the unit propagation, we search in
$\textit{negativeLiteralsToClauses}[\ell]$ for unit clauses, and when
we find one ($\ell'$), we call $\textit{setLiteral}$ again to set the
unset literal to \textit{true}. 

We use the following two lemmas (formally verified by Dafny) to show
that this is sound.

\begin{lemma}\label{unitClauseLiteralFalse_tauNotSatisfiable}
For a truth assignment \textit{tau} and a unit clause $c$ where the
literal $\ell$ is not set, $\textit{tau}[\ell := \textit{false}]$ does
not satisfy the formula.
\end{lemma}
\begin{proof}
  We assume that a complete $tau'$ that extends $\textit{tau}[\ell :=
    \textit{false}]$ and satisfies the formula exists. But all
  literals in $c$ evaluate to \textit{false} under \textit{tau} and
  therefore under $\textit{tau}'$ as well. The truth assignment
  $\textit{tau}'$ does not satisfy clause $c$, and therefore does not
  satisfy the formula either, resulting in a contradiction.
\end{proof}

\begin{lemma}\label{forVariableNotSatisfiableExtend_notSatisfiableExtend}
Given a truth assignment \textit{tau}, if for a literal $\ell$,
$\textit{tau}[\ell := \textit{false}]$ and $\textit{tau}[\ell :=
  \textit{true}]$ do not satisfy the formula when extended, then
\textit{tau} does not satisfy the formula either.
\end{lemma}
\begin{proof}
  Let us assume that \textit{tau} if extended could satisfy the
  formula, therefore there exists a complete extension $\textit{tau}'$
  that satisfies the formula. But $tau'[\ell]$ must be \textit{true}
  or \textit{false}, which makes it an extension of $\textit{tau}[\ell
    := \textit{true}/\textit{false}]$ which can not be extend to
  satisfy the formula, resulting in a contradiction.
\end{proof}

Setting a variable and performing unit propagation is separated as
method $\textit{step}(\ell, \textit{value})$, in order to be make the
development more modular and therefore easier to prove. Most of the
pre- and post-conditions are exactly the same as the ones in
\textit{solve}. 

\section{Benchmarks}
\label{sec:benchmarks}

Dafny code can be extracted to \Csh{} and compiled. We used a few
tests from the latest SAT competitions and we ran them side by side
against the SAT Solver MiniSat\footnote{\url{http://minisat.se/}}.

For experimenting, we restricted our tests only to the tests presented
in Table~\ref{tab:timme}, which were collected by Gregory J. Duck and
published on his
website\footnote{\url{https://www.comp.nus.edu.sg/\~gregory/sat/}}. The tests
were ran on an Intel Core i5-8250U 3.40GHz with 8GB of RAM, Operating
System 5.1.14-arch1-1-ARCH, Dafny 2.3.0.10506, GCC 9.1.0.

\begin{table}[th]
    \centering
    \begin{tabular}{ | p{4cm} | p{4cm} | p{4cm} | } 
    \hline
    Test (number of variables, number of clauses) & Our verified Dafny \newline SAT Solver & MiniSat v2.2.0 \\
    \hline
    Hole6 (42, 133) & UNSAT / 3.21s & UNSAT / 0.02s \\
    Zebra (155, 1135) & SAT / 1.09s & SAT / 0.00s \\
    Hanoi4 (718, 4934) & timed out & SAT / 0.03s \\
    Queens16 (256, 6336) & SAT / 4.64s & SAT / 0.00s \\
    \hline
    \end{tabular}
    \caption{\label{tab:timme}Response and the time required to solve the tests}
\end{table}

As expected, a CDCL solver outperforms our solver. However, the
soundness guarantee offered by our verifier solver is higher than the
unverified C code of MiniSat.

\section{Conclusion and Further Work}
\label{sec:conclusion}

We have developed a formally verified implementation of the DPLL
algorithm in the Dafny programming language. Our implementation
incorporates data structures to quickly identify unit clauses and
perform unit propagation. However, it uses a fixed variable ordering.

The formalization consists of about 3200 lines of Dafny code. The
project was developed in a year of part-time work, as part of the BsC
thesis of the first author. The code was written by the first author,
who also lerned Dafny during that time. The ratio between lines of
proofs and code is 5 to 1. The function \textit{undoLayerOnStack} has
the biggest proof-to-code ratio: 27 lines of actual code and 280 lines
of annotations. The entire Dafny implementation of the solver is
available at \url{https://github.com/andricicezar/sat-solver-dafny}.

The solver is not currently competitive against state-of-the-art CDCL
solvers, but since Dafny compiles to \Csh, we conjecture that it is
possible in theory to obtain performance close to the state-of-the-art
by implementating the rest of the optimizations present in CDCL. We
base our conjecture on a quick test that shows that Dafny code for
enumerating permutations is roughly as performant as hand-written \Csh
code for the same task.


Development of the verified solver was challenging, since the
verification time is prohibitive for certain methods. Here is a
summary of the verification time required for the various methods
implemented as part of the solver:

\begin{table}[th]
    \centering
  \begin{tabular}{ | p{6cm} | p{6cm} | } 
    \hline
    Function / Method / Lemma & Time (seconds)\\
    \hline
    Formula.setLiteral($\ell, \textit{value}$) & 630.59 \\
    SATSolver.solve() & 332.01 \\
    Formula.setVariable($v, \textit{value}$) & 294.54 \\
    Formula.undoLayerOnStack() & 249.16 \\
    SATSolver.step() & 233.25 \\
    Formula.constructor($\textit{vC}, \textit{clauses}$) & 91.14 \\
    Other 32 & Less than 3 seconds each \\
    \hline
  \end{tabular}
  \caption{Time required to prove each method / lemma in seconds}
\end{table}

As further work, the main priority is to lower the verification
time. This will enable experimenting with additional SAT solving
optimizations. Another possible extension is to port the data
structures and algorithms, together with their invariants, to a
verifier for C, such as VCC or Frama-C. This should give C-like
performance to the verified SAT solver.

\paragraph{Acknowledgement.} This work was supported by a grant of the \emph{Alexandru Ioan Cuza University}
of Ia\c{s}i within the Research Grants program \emph{Grant UAIC}, code
GI-UAIC-2018-07.

\nocite{*}
\bibliographystyle{eptcs}

\end{document}

%% file: example-figure.tex
\tikzstyle{clause} = [rectangle, text width=4cm, minimum height=0.75cm]

\tikzstyle{layer} = [rectangle, text width=6cm, minimum height=0.75cm, draw=black]

\begin{tikzpicture}[node distance=0.75cm]
\tikzstyle{every node}=[font=\small]
\node (clause1) [clause] {1) $\textcolor{green}{x_{1}} \lor \textcolor{red}{x_{2}} \lor \textcolor{red}{x_{3}}$};
\node (clause2) [clause, below of=clause1] {2) $\textcolor{red}{\lnot x_{1}} \lor \textcolor{green}{\lnot x_{2}}$};
\node (clause3) [clause, below of=clause2] {3) $\textcolor{red}{x_{2}} \lor \textcolor{green}{\lnot x_{3}}$};
\node (clause4) [clause, below of=clause3] {4) $\textcolor{red}{x_{2}} \lor \textcolor{green}{x_{4}} \lor x_{5}$};
\node (clause5) [clause, below of=clause4] {5) $\textcolor{green}{x_{5}} \lor x_{6} \lor x_{7}$};

\node (layer1) [layer, right of=clause1, xshift=5cm] {$(x_{1}, true)$, $(x_{2}, false)$, $(x_{3}, false)$};
\node (layer2) [layer, below of=layer1] {$(x_{4}, true)$};
\node (layer3) [layer, below of=layer2] {$(x_{5}, true)$};

\node (text1) [clause, above of=layer1, text width=6cm] {Stack:};
\node (text2) [clause, above of=clause1] {Formula:};

\end{tikzpicture}

%% file: solve-flowchart.tex
\tikzstyle{startstop} = [rectangle, rounded corners, minimum width=1cm, minimum height=0.75cm,text centered, draw=black]
\tikzstyle{command} = [rectangle, minimum height=0.75cm,text centered, draw=black ]
\tikzstyle{condition} = [rectangle, text width=3cm, text centered, draw=black]
\tikzstyle{conditionanswer} = [circle, text centered, draw=black]
\tikzstyle{proposition} = [rectangle, minimum height=0.75cm, text width=8cm]
\tikzstyle{arrow} = [thick,->,>=stealth]
\tikzstyle{arrow2} = [dashed,-,>=stealth]

\begin{tikzpicture}[node distance=1.5cm]
\tikzstyle{every node}=[font=\small]
\node (start) [startstop] {solve()};
\node (qhasEmptyClauses) [condition, below of=start] {does the formula have an empty clause?}; 
\draw [arrow] (start) -- (qhasEmptyClauses);
\node (qhasEmptyClausesNo) [conditionanswer, below of=qhasEmptyClauses] {No};
 \draw [arrow] (qhasEmptyClauses) -- (qhasEmptyClausesNo);
\node (qhasEmptyClausesYes) [conditionanswer, right of=qhasEmptyClauses, xshift=2cm] {Yes};
\draw [arrow] (qhasEmptyClauses) -- (qhasEmptyClausesYes);
\node (firstunsat) [startstop, below of=qhasEmptyClausesYes] {UNSAT};
\draw [arrow] (qhasEmptyClausesYes) -- (firstunsat);

\node (qisFormulaEmpty) [condition, below of=qhasEmptyClausesNo] {is formula empty?};
\draw [arrow] (qhasEmptyClausesNo) -- (qisFormulaEmpty);
\node (qisFormulaEmptyNo) [conditionanswer, below of=qisFormulaEmpty] {No};
\draw [arrow] (qisFormulaEmpty) -- (qisFormulaEmptyNo);
\node (qisFormulaEmptyYes) [conditionanswer, right of=qisFormulaEmpty, xshift=1.9cm] {Yes};
\draw [arrow] (qisFormulaEmpty) -- (qisFormulaEmptyYes);
\node (firstsat) [startstop, below of=qisFormulaEmptyYes] {SAT};
\draw [arrow] (qisFormulaEmptyYes) -- (firstsat);

\node (lchooseLiteral) [command, below of=qisFormulaEmptyNo] {$\ell \leftarrow chooseLiteral()$};
\draw [arrow] (qisFormulaEmptyNo) -- (lchooseLiteral);

\node (rstepltrue) [command, below of=lchooseLiteral] {$r \leftarrow step(\ell, true)$};
\draw [arrow] (lchooseLiteral) -- (rstepltrue);

\node (qisrsat) [condition, below of=rstepltrue] {is r SAT?};
\draw [arrow] (rstepltrue) -- (qisrsat);
\node (qisrsatNo) [conditionanswer, below of=qisrsat] {No};
\draw [arrow] (qisrsat) -- (qisrsatNo);
\node (qisrsatYes) [conditionanswer, right of=qisrsat, xshift=1.9cm] {Yes};
\draw [arrow] (qisrsat) -- (qisrsatYes);
\node (scndsat) [startstop, below of=qisrsatYes] {SAT};
\draw [arrow] (qisrsatYes) -- (scndsat);

\node (returnn) [command, below of=qisrsatNo] {return $step(\ell, false)$};
\draw [arrow] (qisrsatNo) -- (returnn);

\node (propunsat) [proposition, right of=firstunsat, yshift=0.75cm, xshift=4.5cm] {$\exists i\ \bullet 0\ \leq i\ < |clauses|\ \land falseLiteralsCount[ i] \ =|clause[ i] |\ \Longrightarrow !isSatisfiableExtend(tau)$};
\draw [arrow2] (propunsat) -- (3.5, -2.25);

\node (prop1) [proposition, right of=qhasEmptyClausesNo, yshift=-0.75cm, xshift=8cm] {1) $\forall i \bullet 0 \leq i < |clauses| \Longrightarrow falseLiteralsCount[i] < literalsCount[i];$};
\draw [arrow2] (prop1) -- (0, -3.67);

\node (propfstsat) [proposition, right of=firstsat, yshift=0.75cm, xshift=4.5cm] {$\forall i \bullet 0 \leq i < |clauses| \Longrightarrow trueLiteralsCount[i] > 0$};
\draw [arrow2] (propfstsat) -- (3.5, -5.25);

\node (prop2) [proposition, right of=qisFormulaEmptyNo, yshift=-0.75cm, xshift=8cm] {2) (\textcircled{1} $ \land \exists i \bullet 0 \leq i < |clauses| \Longrightarrow trueLiteralCount[i] = 0) \Longrightarrow \exists \ell \bullet getLiteralValue(tau, \ell) = -1$};
\draw [arrow2] (prop2) -- (0, -6.67);

\node (prop3) [proposition, right of=lchooseLiteral, yshift=-0.75cm, xshift=8cm] {3) \textcircled{2} $ \land  getLiteralValue(tau, \ell) = -1$};
\draw [arrow2] (prop3) -- (0, -8.33);

\node (prop4) [proposition, right of=rstepltrue, yshift=-0.75cm, xshift=8cm] {4) \textcircled{3} $ \land  old(state) = state$};
\draw [arrow2] (prop4) -- (0, -9.67);

\node (propsndsat) [proposition, right of=scndsat, yshift=0.75cm, xshift=4.5cm] {\textcircled{4} $ \land isSatisfiableExtend(tau[\ell := true]) \Longrightarrow isSatisfiableExtend(tau)$};
\draw [arrow2] (propsndsat) -- (3.5, -11.34);

\node (propfinal) [proposition, right of=qisrsatNo, yshift=-0.75cm, xshift=8cm] {\textcircled{4} $ \land \lnot isSatisfiableExtend(tau[\ell := true])$};
\draw [arrow2] (propfinal) -- (0, -12.67);
\end{tikzpicture}

%% file: main.bbl
\begin{thebibliography}{10}
\providecommand{\bibitemdeclare}[2]{}
\providecommand{\surnamestart}{}
\providecommand{\surnameend}{}
\providecommand{\urlprefix}{Available at }
\providecommand{\url}[1]{\texttt{#1}}
\providecommand{\href}[2]{\texttt{#2}}
\providecommand{\urlalt}[2]{\href{#1}{#2}}
\providecommand{\doi}[1]{doi:\urlalt{http://dx.doi.org/#1}{#1}}
\providecommand{\bibinfo}[2]{#2}

\bibitemdeclare{inproceedings}{DBLP:conf/aaai/BalyoHJ17}
\bibitem{DBLP:conf/aaai/BalyoHJ17}
\bibinfo{author}{Tom{\'{a}}s \surnamestart Balyo\surnameend},
  \bibinfo{author}{Marijn J.~H. \surnamestart Heule\surnameend} \&
  \bibinfo{author}{Matti \surnamestart J{\"{a}}rvisalo\surnameend}
  (\bibinfo{year}{2017}): \emph{\bibinfo{title}{{SAT} Competition 2016: Recent
  Developments}}.
\newblock In: {\sl \bibinfo{booktitle}{Proceedings of the Thirty-First {AAAI}
  Conference on Artificial Intelligence, February 4-9, 2017, San Francisco,
  California, {USA.}}}, pp. \bibinfo{pages}{5061--5063}.

\bibitemdeclare{article}{DBLP:journals/jar/BlanchetteFLW18}
\bibitem{DBLP:journals/jar/BlanchetteFLW18}
\bibinfo{author}{Jasmin~Christian \surnamestart Blanchette\surnameend},
  \bibinfo{author}{Mathias \surnamestart Fleury\surnameend},
  \bibinfo{author}{Peter \surnamestart Lammich\surnameend} \&
  \bibinfo{author}{Christoph \surnamestart Weidenbach\surnameend}
  (\bibinfo{year}{2018}): \emph{\bibinfo{title}{A Verified {SAT} Solver
  Framework with Learn, Forget, Restart, and Incrementality}}.
\newblock {\sl \bibinfo{journal}{J. Autom. Reasoning}}
  \bibinfo{volume}{61}(\bibinfo{number}{1-4}), pp. \bibinfo{pages}{333--365},
  \doi{10.1007/s10817-018-9455-7}.

\bibitemdeclare{inproceedings}{DBLP:conf/sat/BrummayerLB10}
\bibitem{DBLP:conf/sat/BrummayerLB10}
\bibinfo{author}{Robert \surnamestart Brummayer\surnameend},
  \bibinfo{author}{Florian \surnamestart Lonsing\surnameend} \&
  \bibinfo{author}{Armin \surnamestart Biere\surnameend}
  (\bibinfo{year}{2010}): \emph{\bibinfo{title}{Automated Testing and Debugging
  of {SAT} and {QBF} Solvers}}.
\newblock In: {\sl \bibinfo{booktitle}{Theory and Applications of
  Satisfiability Testing - {SAT} 2010, 13th International Conference, {SAT}
  2010, Edinburgh, UK, July 11-14, 2010. Proceedings}}, pp.
  \bibinfo{pages}{44--57}, \doi{10.1007/978-3-642-14186-7\_6}.

\bibitemdeclare{article}{DBLP:journals/jsat/CokDW14}
\bibitem{DBLP:journals/jsat/CokDW14}
\bibinfo{author}{David~R. \surnamestart Cok\surnameend}, \bibinfo{author}{David
  \surnamestart D{\'{e}}harbe\surnameend} \& \bibinfo{author}{Tjark
  \surnamestart Weber\surnameend} (\bibinfo{year}{2014}):
  \emph{\bibinfo{title}{The 2014 {SMT} Competition}}.
\newblock {\sl \bibinfo{journal}{{JSAT}}} \bibinfo{volume}{9}, pp.
  \bibinfo{pages}{207--242}.

\bibitemdeclare{inproceedings}{DBLP:conf/nfm/Fleury19}
\bibitem{DBLP:conf/nfm/Fleury19}
\bibinfo{author}{Mathias \surnamestart Fleury\surnameend}
  (\bibinfo{year}{2019}): \emph{\bibinfo{title}{Optimizing a Verified {SAT}
  Solver}}.
\newblock In: {\sl \bibinfo{booktitle}{{NASA} Formal Methods - 11th
  International Symposium, {NFM} 2019, Houston, TX, USA, May 7-9, 2019,
  Proceedings}}, pp. \bibinfo{pages}{148--165},
  \doi{10.1007/978-3-030-20652-9\_10}.

\bibitemdeclare{article}{dafnyReferenceManual}
\bibitem{dafnyReferenceManual}
\bibinfo{author}{Richard~L. \surnamestart Ford\surnameend} \&
  \bibinfo{author}{K.~Rustan~M. \surnamestart Leino\surnameend}
  (\bibinfo{year}{2017}): \emph{\bibinfo{title}{Dafny Reference Manual}}.

\bibitemdeclare{incollection}{DBLP:reference/fai/GomesKSS08}
\bibitem{DBLP:reference/fai/GomesKSS08}
\bibinfo{author}{Carla~P. \surnamestart Gomes\surnameend},
  \bibinfo{author}{Henry~A. \surnamestart Kautz\surnameend},
  \bibinfo{author}{Ashish \surnamestart Sabharwal\surnameend} \&
  \bibinfo{author}{Bart \surnamestart Selman\surnameend}
  (\bibinfo{year}{2008}): \emph{\bibinfo{title}{Satisfiability Solvers}}.
\newblock In: {\sl \bibinfo{booktitle}{Handbook of Knowledge Representation}},
  \bibinfo{publisher}{Elsevier}, pp. \bibinfo{pages}{89--134},
  \doi{10.1016/S1574-6526(07)03002-7}.

\bibitemdeclare{inproceedings}{sel4}
\bibitem{sel4}
\bibinfo{author}{Gerwin \surnamestart Klein\surnameend}, \bibinfo{author}{Kevin
  \surnamestart Elphinstone\surnameend}, \bibinfo{author}{Gernot \surnamestart
  Heiser\surnameend}, \bibinfo{author}{June \surnamestart
  Andronick\surnameend}, \bibinfo{author}{David \surnamestart Cock\surnameend},
  \bibinfo{author}{Philip \surnamestart Derrin\surnameend},
  \bibinfo{author}{Dhammika \surnamestart Elkaduwe\surnameend},
  \bibinfo{author}{Kai \surnamestart Engelhardt\surnameend},
  \bibinfo{author}{Rafal \surnamestart Kolanski\surnameend},
  \bibinfo{author}{Michael \surnamestart Norrish\surnameend},
  \bibinfo{author}{Thomas \surnamestart Sewell\surnameend},
  \bibinfo{author}{Harvey \surnamestart Tuch\surnameend} \&
  \bibinfo{author}{Simon \surnamestart Winwood\surnameend}
  (\bibinfo{year}{2009}): \emph{\bibinfo{title}{seL4: formal verification of an
  {OS} kernel}}.
\newblock In: {\sl \bibinfo{booktitle}{Proceedings of the 22nd {ACM} Symposium
  on Operating Systems Principles 2009, {SOSP} 2009, Big Sky, Montana, USA,
  October 11-14, 2009}}, pp. \bibinfo{pages}{207--220},
  \doi{10.1145/1629575.1629596}.

\bibitemdeclare{article}{compcert}
\bibitem{compcert}
\bibinfo{author}{Xavier \surnamestart Leroy\surnameend} (\bibinfo{year}{2009}):
  \emph{\bibinfo{title}{Formal verification of a realistic compiler}}.
\newblock {\sl \bibinfo{journal}{Commun. {ACM}}}
  \bibinfo{volume}{52}(\bibinfo{number}{7}), pp. \bibinfo{pages}{107--115},
  \doi{10.1145/1538788.1538814}.

\bibitemdeclare{phdthesis}{lescuyer:tel-00713668}
\bibitem{lescuyer:tel-00713668}
\bibinfo{author}{Stephane \surnamestart Lescuyer\surnameend}
  (\bibinfo{year}{2011}): \emph{\bibinfo{title}{Formalizing and Implementing a
  Reflexive Tactic for Automated Deduction in Coq}}.
\newblock \bibinfo{type}{Theses}, \bibinfo{school}{{Universit{\'e} Paris Sud -
  Paris XI}}.

\bibitemdeclare{article}{DBLP:journals/jar/Maric09}
\bibitem{DBLP:journals/jar/Maric09}
\bibinfo{author}{Filip \surnamestart Mari{\'{c}}\surnameend}
  (\bibinfo{year}{2009}): \emph{\bibinfo{title}{Formalization and
  Implementation of Modern {SAT} Solvers}}.
\newblock {\sl \bibinfo{journal}{J. Autom. Reasoning}}
  \bibinfo{volume}{43}(\bibinfo{number}{1}), pp. \bibinfo{pages}{81--119},
  \doi{10.1007/s10817-009-9127-8}.

\bibitemdeclare{inproceedings}{DBLP:conf/vmcai/OeSOC12}
\bibitem{DBLP:conf/vmcai/OeSOC12}
\bibinfo{author}{Duckki \surnamestart Oe\surnameend}, \bibinfo{author}{Aaron
  \surnamestart Stump\surnameend}, \bibinfo{author}{Corey \surnamestart
  Oliver\surnameend} \& \bibinfo{author}{Kevin \surnamestart Clancy\surnameend}
  (\bibinfo{year}{2012}): \emph{\bibinfo{title}{versat: {A} Verified Modern
  {SAT} Solver}}.
\newblock In: {\sl \bibinfo{booktitle}{Verification, Model Checking, and
  Abstract Interpretation - 13th International Conference, {VMCAI} 2012,
  Philadelphia, PA, USA, January 22-24, 2012. Proceedings}}, pp.
  \bibinfo{pages}{363--378}, \doi{10.1007/978-3-642-27940-9\_24}.

\bibitemdeclare{article}{DBLP:journals/entcs/ShankarV11}
\bibitem{DBLP:journals/entcs/ShankarV11}
\bibinfo{author}{Natarajan \surnamestart Shankar\surnameend} \&
  \bibinfo{author}{Marc \surnamestart Vaucher\surnameend}
  (\bibinfo{year}{2011}): \emph{\bibinfo{title}{The Mechanical Verification of
  a DPLL-Based Satisfiability Solver}}.
\newblock {\sl \bibinfo{journal}{Electr. Notes Theor. Comput. Sci.}}
  \bibinfo{volume}{269}, pp. \bibinfo{pages}{3--17},
  \doi{10.1016/j.entcs.2011.03.002}.

\bibitemdeclare{article}{DBLP:journals/aicom/Sutcliffe18}
\bibitem{DBLP:journals/aicom/Sutcliffe18}
\bibinfo{author}{Geoff \surnamestart Sutcliffe\surnameend}
  (\bibinfo{year}{2018}): \emph{\bibinfo{title}{The 9th {IJCAR} Automated
  Theorem Proving System Competition - {CASC-J9}}}.
\newblock {\sl \bibinfo{journal}{{AI} Commun.}}
  \bibinfo{volume}{31}(\bibinfo{number}{6}), pp. \bibinfo{pages}{495--507},
  \doi{10.3233/AIC-180773}.

\end{thebibliography}
